\documentclass[a4paper,12pt]{article}
\usepackage{amsmath,amssymb,amsfonts,amsthm}
\newtheorem{theorem}{Theorem}[section]

\newtheorem{definition}[theorem]{Definition}

\usepackage{wasysym}
\usepackage{graphicx}
\newcommand{\R}{\mathcal{R}}

\newcommand{\Rt}{\mathbb{R}^3}

\newcommand{\dom}{U}
\newcommand{\Su}{\mathcal{S}}

\newcommand{\E}{\mathcal{E}}
\newcommand{\EH}{\mathcal{E}_H}

\title{Bekenstein bounds and inequalities between size, charge,  angular momentum
  and energy for bodies}
\author{Sergio Dain\\
  Facultad de Matem\'atica, Astronom\'{i}a y F\'{i}sica, \\
     Universidad Nacional de C\'ordoba, \\
Instituto de F\'{i}sica Enrique Gaviola, IFEG, CONICET,\\
  Ciudad Universitaria (5000) C\'ordoba, Argentina.}

\begin{document}
\maketitle

\begin{abstract}
  Bekenstein bounds for the entropy of a body imply a universal inequality
  between size, energy, angular momentum and charge. We prove this inequality
  in Electromagnetism. We also prove it, for the particular case of zero
  angular momentum, in General Relativity. We further discuss the relation of
  these inequalities with inequalities between size, angular momentum and
  charge recently studied in the literature.
\end{abstract}

\section{Introduction}
\label{sec:bekenstein-bounds}

A universal bound on the entropy of a macroscopic body has been proposed by
Bekenstein \cite{Bekenstein:1980jp}
\begin{equation}
  \label{eq:3}
 \frac{\hbar c}{2\pi k_B} S\leq  \E \R
\end{equation}
where $S$ is the entropy, $k_B$ is Boltzmann's constant, $\R$ is the radius of
the smallest sphere that can enclose the body, $\E$ is the total energy,
$\hbar$ is the reduced Planck constant, and $c$ is the speed of light.  Using a
similar kind of heuristic arguments, a generalization of (\ref{eq:3}) including
the electric charge $Q$ and the angular momentum $J$ of the body has been also
proposed \cite{Hod:2000ju} \cite{Zaslavskii92} \cite{Bekenstein:1999cf}
\begin{equation}
  \label{eq:6}
\frac{\hbar c}{2\pi k_B}  S\leq \sqrt{(\E\R)^2-c^2J^2}-\frac{Q^2}{2}.
\end{equation}
The original physical arguments used to present these inequalities
involve black holes.  However, a remarkable feature of these inequalities is
that the gravitational constant $G$ does not appear in them.

The bound (\ref{eq:3}) has been extensively studied, see, for example, the
review articles \cite{Bekenstein:2004sh}, \cite{Bousso:2002ju}, \cite{Wald01}
and references therein.  However, the generalization (\ref{eq:6}) appears to
have received much less attention.  In particular, since the entropy $S$ is
always non-negative, the bound (\ref{eq:6}) implies the following inequality in
which the entropy $S$ and the constant $\hbar$ are not involved
\begin{equation}
  \label{eq:76}
  \frac{Q^4}{4\R^2} + \frac{c^2 J^2}{\R^2}\leq \E^2.
\end{equation}
Equality in \eqref{eq:76} implies, by (\ref{eq:6}), that the entropy of the
body is zero and hence the system should be in a very particular state.  Then,
we expect some kind of rigidity statement for the equality in (\ref{eq:76}). 

The main purpose of this article is to study  inequality (\ref{eq:76}).  The
only fundamental constant that appears in \eqref{eq:76} is $c$. Hence, the
obvious theory to test  inequality (\ref{eq:76}) is Electromagnetism. To the
best of our knowledge, such basic study, in full generality, has not been done
before.  In section \ref{sec:electromagnetism} we prove that \eqref{eq:76}
holds as a consequence of Maxwell equations. This theorem provides an indirect
but highly non-trivial evidence in favor of the bound (\ref{eq:6}).

In section \ref{sec:general-relativity} we first discuss the relation of the
inequality \eqref{eq:76} with inequalities between size, angular momentum and
charge recently studied in General Relativity \cite{Dain:2013gma}. Then, we
point out that a result of Reiris \cite{Reiris:2014tva} proves  inequality
\eqref{eq:76} in spherical symmetry in General Relativity.  Finally, we
generalize this result and prove  inequality \eqref{eq:76}, with $J=0$, for
time-symmetric initial data.

\section{Electromagnetism}
\label{sec:electromagnetism}
To fix the notation, let us write Maxwell's equations in Gaussian units
\begin{align}
  \label{eq:32}
  \nabla \times \mathbf{B}- \frac{1}{c}\frac{\partial \mathbf{E}}{\partial t} &=\frac{4\pi}{c}\mathbf{j},  &  \nabla \cdot  \mathbf{E} &= 4\pi \rho, \\
  \nabla \times \mathbf{E}+\frac{1}{c}\frac{\partial \mathbf{B}}{\partial t} &=0, & \nabla \cdot \mathbf{B} &=0,\label{eq:32v}
\end{align}
where $\mathbf{E}$, $\mathbf{B}$ are the electric and magnetic field, and $
\rho$, $\mathbf{j}$ are the charge and current density. These equations are
written in terms of inertial coordinates $(t, \mathbf{x})$ where $t$ is the
time coordinate and $ \mathbf{x}$ are spatial coordinates centered at an
arbitrary point $x_0$. 

Let $\dom$ be an arbitrary region in  space. The electric charge contained in
$\dom$ is given by
\begin{equation}
  \label{eq:16}
  Q(\dom)=\int_\dom \rho  ,
\end{equation}
and the energy of the electromagnetic field in $\dom$ is 
\begin{equation}
  \label{eq:26}
  \E(\dom)= \frac{1}{8\pi} \int_{\dom} |\mathbf{E}|^2 + |\mathbf{B}|^2  \, .
\end{equation}

The angular momentum in the region $\dom$ in the direction of the unit vector
$\mathbf{k}$ with respect to the point $x_0$ is given by
\begin{equation}
  \label{eq:24}
  \mathbf{J}\cdot \mathbf{k}=\frac{1}{4\pi c} \int_\dom (\mathbf{x} \times (\mathbf{E} \times
  \mathbf{B}))\cdot \mathbf{k}  \, .
\end{equation}
Finally, in order to study  inequality (\ref{eq:76}) we need to provide a
definition of the radius $\R$ for an arbitrary region $\dom$. 

\begin{definition}
\label{d:R}
We define the radius $\R$ of the region $\dom$ as the the radius of the
smallest sphere  that encloses  $\dom$. 
 \end{definition}
 Given a domain $\dom$, we denote by $B_{\R}$ the smallest ball  that
 encloses $\dom$ and $x_0$ denote the center of this ball. Note that, in general,
 $x_0$ is not in $\dom$, see figure \ref{fig:1}. We denote by $\partial B_{\R}$
 the boundary of $B_\R$, that is, the sphere of radius $\R$ centered at $x_0$. 
\begin{figure}
  \centering
  \includegraphics{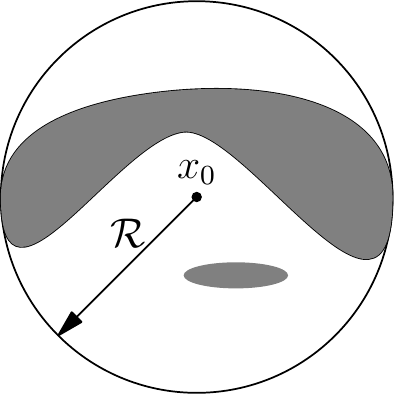}
  \caption{The domain $\dom$ is colored with gray. The radius $\R$ is defined
    as the radius of  the smallest sphere that
    encloses $\dom$. For this particular domain $\dom$ the center $x_0$ of that
    sphere is not in $\dom$.}
  \label{fig:1}
\end{figure}

Before dealing with the general case, it is useful to begin with
Electrostatics, which in particular implies $J=0$.  We will see that the proof
for the dynamical case is based on the proof for the Electrostatics case. Also,
in Electrostatics is simpler to discuss the scope of inequality (\ref{eq:76}).

The equations of Electrostatics  are given by
\begin{equation}
  \label{eq:14}
  \nabla \cdot \mathbf{E} = 4 \pi \rho, \quad \nabla \times \mathbf{E} =0.
\end{equation}

The potential $\Phi$ is defined by $\mathbf{E}=-\nabla \Phi$ and it satisfies
the Poisson equation
\begin{equation}
  \label{eq:17}
  \Delta \Phi= -4\pi \rho. 
\end{equation}
Using equation (\ref{eq:17}) and Gauss's  theorem we obtain that the charge can be
written as boundary integral
\begin{equation}
  \label{eq:64}
  Q(\dom)=-\frac{1}{4\pi}\oint_{\partial \dom} \partial_n \Phi ,
\end{equation}
where $\partial_n$ denote partial derivative along  the exterior unit
normal vector of the boundary $\partial \dom$. 
The total  Electrostatics energy is given by
\begin{equation}
  \label{eq:15}
  \E=\frac{1}{8\pi}\int_{\Rt} |\mathbf{E}|^2  .
\end{equation}

\begin{theorem}
\label{t:1}
Assume that the charge density $\rho$ has compact support contained in the
region $\dom$.  In Electrostatics (i.e. we assume equations (\ref{eq:14})), the
following inequality holds
  \begin{equation}
    \label{eq:5}
    Q^2\leq 2 \E \R,
  \end{equation}
  where $Q$ is the charge contained in $\dom$, $\R$ is the radius of $\dom$
  defined above and $\E$ is the total electromagnetic energy given by
  (\ref{eq:15}). The equality in (\ref{eq:5}) holds if and only if the electric
  field is equal to the electric field produced by an spherical thin shell of
  constant surface charge density and radius $\R$. In particular, this implies that the
  electric field vanished  inside $\dom$. 
\end{theorem}
\begin{proof}
  The system has electric field $\mathbf{E}$ (with potential
  $\Phi$), charge density $\rho$ with support in $\dom$ and total charge
  $Q$. Let $\R$ be the radius of the domain $\dom$ defined in \ref{d:R} and
  $B_{\R}$ its corresponding ball centered at $x_0$.

 Consider the following auxiliary potential defined by
 \begin{equation}
  \label{eq:11}
  \Phi_0=\begin{cases}\frac{Q}{r} \text{ if } r\geq \R,\\
\frac{Q}{\R} \text{ if } r\leq \R,
\end{cases}
\end{equation}
where $r$ is the radial distance to $x_0$. The potential $\Phi_0$ corresponds
to the potential of an spherical thin shell of radius $\R$, constant surface charge
density and total charge $Q$. 

Define $\Phi_1$ by the difference
\begin{equation}
  \label{eq:58}
  \Phi_1=\Phi-\Phi_0.
\end{equation}
By construction $\Phi_1$ satisfies
\begin{equation}
  \label{eq:59}
  \Delta \Phi_1=\begin{cases} 0  &\text{ if } r> \R,\\
-4\pi \rho &\text{ if } r < \R,
\end{cases}
\end{equation}
and
\begin{equation}
  \label{eq:60}
  \oint_{\partial B_\R} \partial_r \Phi_1 =0.
\end{equation}
Equation (\ref{eq:60}) follows since in the definition of $\Phi_0$ we have used
the total charge $Q$ of the potential $\Phi$.

The total energy of the system  is given by
\begin{align}
  \label{eq:61}
  \E & =\frac{1}{8\pi} \int_{\Rt} |\nabla \Phi|^2,\\
     & = \frac{1}{8\pi} \int_{\Rt} |\nabla \Phi_0|^2+ |\nabla \Phi_1|^2+2 \nabla
       \Phi_0 \cdot \nabla \Phi_1,\label{eq:61b}
\end{align}
where in line (\ref{eq:61b}) we have used the definition (\ref{eq:58}). To
calculate the last term in (\ref{eq:61b}) we decompose the domain of the  integral in
$\Rt\setminus B_\R$ and $B_\R$. We have
\begin{equation}
  \label{eq:62}
  \int_{B_\R}   \nabla \Phi_0 \cdot \nabla \Phi_1 = 0,
\end{equation}
since $\Phi_0$ is constant in $B_R$. For the other integral we have
\begin{equation}
  \label{eq:63}
  \int_{\Rt\setminus B_\R}   \nabla \Phi_0 \cdot \nabla \Phi_1 =
  \int_{\Rt\setminus B_\R}  \nabla \cdot ( \Phi_0  \nabla \Phi_1) - \Phi_0 \Delta \Phi_1.
\end{equation}
Since $\Delta \Phi_1=0$ on $\Rt\setminus B_R$ the second term in the right
hand side of (\ref{eq:63}) vanishes. The first term can be converted in a
boundary integral
\begin{equation}
  \label{eq:8}
   \int_{\Rt\setminus B_\R}  \nabla ( \Phi_0  \nabla \Phi_1)=\lim_{r\to \infty}
   \oint_{\partial B_r} \Phi_0 \partial_r \Phi_1 -   \oint_{\partial B_\R} \Phi_0 \partial_r \Phi_1.
\end{equation}
The first term on the right hand side of equation (\ref{eq:8}) vanishes by the
decay conditions of $\Phi_0$ and $\Phi_1$. For the second term we have
\begin{align}
  \label{eq:9}
  \oint_{\partial B_\R} \Phi_0 \partial_r \Phi_1 &=\Phi_0 \oint_{\partial
                                                  B_\R} \partial_r \Phi_1,\\
&=0. \label{eq:9bc}
\end{align}
Where in line (\ref{eq:9}) we have used that $\Phi_0$ is constant on spheres
and in line (\ref{eq:9bc}) we have used (\ref{eq:60}). 
Hence, we have proved that
\begin{equation}
  \label{eq:10}
  \E  = \frac{1}{8\pi} \int_{\Rt} |\nabla \Phi_0|^2+ |\nabla \Phi_1|^2\, .
\end{equation}
The first term in (\ref{eq:10}) can be computed explicitly using (\ref{eq:11}). It is the binding
energy of an spherical shell of radius $\R$ with constant charge surface density and total
charge $Q$. Then, we finally obtain
\begin{equation}
  \label{eq:18}
  \E  = \frac{Q^2}{2\R}+\frac{1}{8\pi} \int_{\Rt}  |\nabla \Phi_1|^2.
\end{equation}
This equality proves  inequality (\ref{eq:5}) and also the rigidity statement:
if the equality  in  (\ref{eq:5}) holds, then (\ref{eq:18}) implies  $\nabla \Phi_1=0$
and hence $\mathbf{E}=\nabla \Phi_0$. 
\end{proof}
Note that the equality (\ref{eq:18}) implies the following estimate for the
fields inside the domain $\dom$
\begin{equation}
  \label{eq:12}
  \E  - \frac{Q^2}{2\R}\geq \frac{1}{8\pi} \int_{\dom}  |\mathbf{E}|^2,
\end{equation}
where we have used that in $\dom$ we have $\nabla \Phi_1 =\nabla
\Phi=\mathbf{E}$. 

Let us discuss the scope of inequality (\ref{eq:5}).  The first important
observation is that in inequality (\ref{eq:5}) the energy $\E$ is the total
energy of the system, which in electrostatic is equivalent to the binding
energy. That is, $\E$ represents the work needed to assemble the charge
configuration from infinity. Inequality is clearly false if instead of the
total energy we use the integral of the energy density on the domain $\dom$
given by (\ref{eq:26}): for example, take the spherical shell of radius $\R$
and constant surface charge density. Then, the domain $\dom$ is given by the
ball $B_\R$, but the integral of the energy density over $B_\R$ is zero since
the electric field vanishes in $B_\R$.

Inequality (\ref{eq:5}) is not valid if we consider many disconnected
regions and take $Q$ and $\R$ to be the corresponding radius and charge of only
one region and $\E$ the total energy of the system. The counterexample is the
following. Consider two spherical thin shells of constant surface density with
radius $R_1$ and $R_2$ and total charge $Q$ and $-Q$. The separation between
the centers is $L$, and we assume that they do not overlap, i.e.
$L\geq R_1 + R_2$. The total energy of this system is given by
\begin{equation}
  \label{eq:29}
  \E=\E_1 + \E_2 - \frac{Q^2}{L},
\end{equation}
where the self energy of each shell is given by 
\begin{equation}
  \label{eq:30}
  \E_1=\frac{Q^2}{2R_1}, \quad  \E_2=\frac{Q^2}{2R_2}.
\end{equation}
For a simple way to compute  the third  term  in (\ref{eq:29})
(namely, the interaction energy) see, for example \cite{zangwill2013modern}
p. 75. 
At the contact point $L=R_1+R_2$ we have
\begin{equation}
  \label{eq:1}
  \E-\E_1= \frac{Q^2(R_1-R_2)}{2R_2(R_1+R_2)}.
\end{equation}
Take $R_2>R_1$, then if the shells are close enough to the contact point, from (\ref{eq:1}) we
deduce that
\begin{equation}
  \label{eq:31}
   \E-\E_1<0.
\end{equation}
But then
\begin{equation}
  \label{eq:33}
  \E<\E_1=\frac{Q^2}{2R_1},
\end{equation}
and hence inequality (\ref{eq:5}) is not valid for the shell $R_1$ if we
take  in  (\ref{eq:5}) $\E$ as the total energy and $Q$ and $\R$ as the charge
and radius of the shell. 

An alternative and useful way to prove  inequality (\ref{eq:5}) in
electrostatic is the  following. By Thomson's theorem the electrostatic energy of a
body of fixed shape, size and charge is minimized when its charge $Q$
distributes itself to make the electrostatic potential constant throughout the
body (see, for example, \cite{zangwill2013modern} p. 128). That is, the
original configuration is replaced by a conductor with the same total charge
and size which has less or equal energy. For conductors inequality
(\ref{eq:5}) is related with the capacity of the conductor, defined as
follows.  Consider a conductor $\dom$ and define the potential $\Phi_1$ by
\begin{align}
    \label{eq:18b}
    \Delta \Phi_1 &=0 \text{ in } \Rt\setminus \dom,\\
    \Phi_1 &= 1 \text{ at } \partial \dom, \\
\lim_{r\to \infty} \Phi_1&=0.
  \end{align}
The capacity of $\dom$ is given by
\begin{equation}
  \label{eq:22}
  C=-\frac{1}{4\pi}\oint_{\partial \dom} \partial_n \Phi_1.
\end{equation}
The capacity $C$ satisfies the well-known relation
\begin{equation}
  \label{eq:10b}
  \E = \frac{Q^2}{2C},
\end{equation}
where $\E$ is the total electrostatic energy of the conductor. 
Then, for a conductor, inequality (\ref{eq:5}) is equivalent to
\begin{equation}
  \label{eq:70}
  C\leq \R.
\end{equation}
Since, by Thomson's theorem conductors minimize the energy, we have proved
that inequality (\ref{eq:5}) for general configurations reduces to the
inequality \eqref{eq:70} for conductors.

To prove \eqref{eq:70}  we use the variational characterization of $C$
\begin{equation}
  \label{eq:8b}
  C=\frac{1}{4\pi}\inf_{\Phi \in K} \int_{\Rt\setminus \dom} |\nabla \Phi|^2 ,
\end{equation}
where $K$ is the set of all functions $\Phi$ that decay at infinity and are
equal to $1$ at $\partial \dom$.  Consider the following test function
\begin{equation}
  \label{eq:11b}
  \Phi_R=\begin{cases}\frac{\R}{r} \text{ if } r\geq \R,\\
1 \text{ if } r\leq \R.
\end{cases}
\end{equation}
We have that $\Phi_R \in K$ and hence we can use (\ref{eq:8}) to obtain
\begin{equation}
  \label{eq:9b}
  C\leq \frac{1}{4\pi} \int_{\Rt\setminus B_{\R} }|\nabla \Phi_R|^2 =\R.
\end{equation}

This characterization in terms of the capacity is useful to find interesting
examples and estimates.  In particular it allows to prove the following
relevant statement: inequality \eqref{eq:5} is not valid if we replace the
definition of $\R$ by the area radius, namely
\begin{equation}
  \label{eq:71}
  \R_A= \sqrt{\frac{A}{4\pi}},
\end{equation}
where $A$ is the area of the boundary $\partial \dom$. Note that the area
radius $\R_A$ represents perhaps the simplest definition of radius that can be
directly translated into curved spaces. The following counterexample shows
that even in flat space $\R_A$ is not an appropriate measure of size in our
context.

Consider a prolate conducting ellipsoid with radius $a$ and $b$ with $a>b$. The
capacity of this conductor is given by (see
\cite{landau1984electrodynamics}, p. 22)
\begin{equation}
  \label{eq:72}
  C=\frac{\sqrt{a^2-b^2}}{\cosh^{-1} a/b},
\end{equation}
and the surface area is given by
\begin{equation}
  \label{eq:73}
  A= 2\pi b^2\left(1+\frac{a}{b} \frac{\sin^{-1} e}{e}  \right), \quad e^2=1-\frac{b^2}{a^2}.
\end{equation}
We calculate the dimensionless quotient
\begin{equation}
  \label{eq:74}
  \frac{C}{\R_A}=\frac{\sqrt{2} \sqrt{\frac{a^2}{b^2}-1}}{\left(\cosh^{-1} \frac{a}{b}\right)\sqrt{1+\frac{a}{b} \frac{\sin^{-1} e}{e} }}.
\end{equation}
Note that $C/\R_A$ depends only on the dimensionless parameter $a/b$. We take
the limit  $a/b \to \infty$
\begin{equation}
  \label{eq:75}
 \lim_{a/b\to \infty}  \frac{C}{\R_A} \approx \frac{2}{\sqrt{\pi}}
 \frac{\sqrt{a/b}}{\log(a/b)}\to \infty.
\end{equation}
And hence inequality (\ref{eq:70}) is not satisfied for $\R_A$.

With this example we conclude the study of inequality (\ref{eq:5}) in
Electrostatics. From now on, we will deal with the full Maxwell's equations
(\ref{eq:32})--(\ref{eq:32v}).  As a preliminary step, we prove inequality
(\ref{eq:76}) with $Q=0$ and $J\neq 0$. This particular case will also be used
in the general proof of inequality (\ref{eq:76}). 

\begin{theorem}
\label{t:ejr}
Consider a solution of Maxwell's equations (\ref{eq:32})--(\ref{eq:32v}) in the
domain $\dom$. Let $\R$ be the radius of $\dom$ defined in \ref{d:R} and let
$x_0$ be the center of the corresponding sphere. Then the following inequality
holds
\begin{equation}
  \label{eq:27}
  c|J(\dom)|\leq \R \E(\dom),
\end{equation}
where $J(\dom)$ is the angular momentum of the electromagnetic field given by
(\ref{eq:24}) with respect to the point $x_0$. Moreover, the equality in
(\ref{eq:27}) holds if and only if the electromagnetic field vanishes in $\dom$.
\end{theorem}
Note that inequality (\ref{eq:27}) is purely quasilocal, in contrast with
the previous inequality (\ref{eq:5}): in (\ref{eq:27}) there appears only
quantities defined on the domain $\dom$ and not global quantities like the
total energy $\E$. Of course, since $\E \geq \E(\dom)$, inequality
(\ref{eq:27}) implies the global inequality
\begin{equation}
  \label{eq:44}
   c|J(\dom)|\leq \R \E.
\end{equation}
Moreover, theorem \ref{t:ejr} implies also a rigidity statement for the
inequality (\ref{eq:44}): equality holds if and only if the
electromagnetic field vanishes everywhere. 

\begin{proof}
We estimate the difference
\begin{align}
  \label{eq:19}
 \E(\dom)- \frac{c}{\R}|J(\dom)| &= \frac{1}{8\pi} \int_\dom   |\mathbf{E}|^2 +
 |\mathbf{B}|^2 - \frac{1}{4\pi \R} \left|  \int_\dom (\mathbf{x} \times (\mathbf{E} \times
  \mathbf{B}))\cdot \mathbf{k} \right|,\\
&\geq \frac{1}{8\pi} \int_\dom   |\mathbf{E}|^2 +
 |\mathbf{B}|^2 - \frac{2}{\R} |(\mathbf{x} \times (\mathbf{E} \times
  \mathbf{B}))\cdot \mathbf{k}|.\label{eq:19b}
\end{align}
The integrand of the angular momentum (i.e. the angular momentum density)
satisfies the elementary inequality
\begin{align}
  \label{eq:28}
 | (\mathbf{x} \times (\mathbf{E} \times
  \mathbf{B}))\cdot \mathbf{k} | &\leq  |(\mathbf{x} \times (\mathbf{E} \times
  \mathbf{B}))| |\mathbf{k} |,\\
& =  |(\mathbf{x} \times (\mathbf{E} \times
  \mathbf{B}))| ,\label{eq:28a}\\
& \leq |\mathbf{x} |  |\mathbf{E}| |\mathbf{B}|\label{eq:28b},
\end{align}
where in line (\ref{eq:28}) we used the inequality $|\mathbf{a}\cdot
\mathbf{b}|\leq |\mathbf{a}| |\mathbf{b}|$, in line  (\ref{eq:28a}) we 
used that $\mathbf{k}$ is an unit vector and in line (\ref{eq:28b}) we  used
the inequality $|\mathbf{a}\times
\mathbf{b}|\leq |\mathbf{a}| |\mathbf{b}|$. Using this inequality, we obtain
\begin{equation}
  \label{eq:7a}
 |\mathbf{E}|^2 +
 |\mathbf{B}|^2 - \frac{2}{\R} |(\mathbf{x} \times (\mathbf{E} \times
  \mathbf{B}))\cdot \mathbf{k}|  \geq  |\mathbf{E}|^2 +
 |\mathbf{B}|^2 - 2 \frac{|\mathbf{x} |}{\R}  |\mathbf{E}| |\mathbf{B}|. 
\end{equation}
We write the right hand side of the inequality as follows
\begin{align}
  \label{eq:39}
 &|\mathbf{E}|^2 +
 |\mathbf{B}|^2 - 2 \frac{|\mathbf{x} |}{\R}  |\mathbf{E}| |\mathbf{B}| =\\
&=  |\mathbf{E}|^2 +
 |\mathbf{B}|^2 - \frac{|\mathbf{x} |}{\R}\left(|\mathbf{E}|^2+
  |\mathbf{B}|^2\right) + \frac{|\mathbf{x} |}{\R}\left(|\mathbf{E}|^2+
  |\mathbf{B}|^2\right) - 2 \frac{|\mathbf{x} |}{\R}  |\mathbf{E}|
  |\mathbf{B}|,\\
&=\left(1- \frac{|\mathbf{x} |}{\R}\right) \left(|\mathbf{E}|^2 +
 |\mathbf{B}|^2\right) +\frac{|\mathbf{x} |}{\R}  \left(|\mathbf{E}|-
 |\mathbf{B}|  \right)^2,\\
& \geq \left(1- \frac{|\mathbf{x} |}{\R}\right) \left(|\mathbf{E}|^2 +
 |\mathbf{B}|^2\right).
\end{align}
Collecting these inequalities we arrive to our final result
\begin{equation}
  \label{eq:7}
  \E(\dom)- \frac{c}{\R}|J(\dom)| \geq  \frac{1}{8\pi} \int_\dom \left(1- \frac{|\mathbf{x} |}{\R}\right) \left(|\mathbf{E}|^2 +
 |\mathbf{B}|^2\right). 
\end{equation}
By the definition of $\R$ we have $|\mathbf{x}|\leq \R$ on $\dom$, and hence
the integrand on the right hand side of  inequality (\ref{eq:7}) is
non-negative. This proves (\ref{eq:27}). Moreover, inequality \eqref{eq:7}
also proves the rigidity statement: if equality holds, then the integrand on
the right hand side of  inequality (\ref{eq:7}) should vanish. Then, for
every $\mathbf{x} \in \dom$ that is not on the sphere $\partial B_\R$ we have
that both $\mathbf{E}$ and $\mathbf{B}$ are zero. By continuity, the fields are
also zero on the points on the sphere $\partial B_\R$

\end{proof}

The proof of inequality (\ref{eq:27}) (but not the rigidity statement) can
be directly generalized to any classical field theory. It is a direct
consequence of the dominant energy condition\footnote{I thank to G. Dotti for
  providing me this argument.}.  Let $T_{\mu\nu}$ be the electromagnetic energy
momentum tensor of the theory. The indices $\mu,\nu, \cdots$ are
4-dimensional and we are using signature $(-+++)$.  For example, for electromagnetism we have
\begin{equation}
  \label{eq:4b}
  T_{\mu \nu}= \frac{1}{4\pi}\left(F_{\mu \lambda}
    F_\nu{}^{\lambda}-\frac{1}{4}g_{\mu \nu} F_{\lambda \gamma} F^{\lambda \gamma}  \right),
\end{equation}
where $F_{\mu\nu}$ is the (antisymmetric) electromagnetic field tensor that
satisfies Maxwell's equations. Consider a spacelike surface $\dom$ with normal
$t^\mu$. The energy is given by
\begin{equation}
  \label{eq:25}
  \E= \int_{\dom} T_{\mu\nu} t^\mu t^\nu.  
\end{equation}
Let $\eta^\mu$ be a 
Killing vector field that corresponds to space rotations. The angular momentum
corresponding to the rotation $\eta^\mu$ is given by
\begin{equation}
  \label{eq:13}
  J(\dom)=\frac{1}{c}\int_\dom T_{\mu\nu}t^\mu \eta^\nu. 
\end{equation}
Choosing coordinates such that $x^i$ are spacelike Cartesian coordinates on the
surface $U$ and $t^\mu=(1,0,0,0)$, then the space rotations are characterized by 
\begin{equation}
  \label{eq:23}
  \eta_i = \epsilon_{ijk} k^j x^k,
\end{equation}
where $k$ is an arbitrary constant spacelike unit vector that represent the axis of
rotation and the indices $i,j,k\dots$ are 3-dimensional.  For the case of the
electromagnetism, it is easy to check, using (\ref{eq:4b}), that the definition
(\ref{eq:13}) coincides with (\ref{eq:24}).
 
Assume that $T_{\mu\nu}$ satisfies the dominant energy condition, namely
\begin{equation}
\label{eq:32b}
  T_{\mu\nu} \xi^\mu k^\nu\geq 0,
\end{equation}
for all future directed timelike or null vectors $k^\mu$ and $\xi^\mu$.

Denote by $\eta$ the square norm of $\eta^i$, that is 
$\eta=\eta^i\eta_i=\eta^\mu\eta_\mu$  and define the unit vector
$\hat \eta^\mu = \eta^\mu \eta^{-1/2}$. Then, the vector
\begin{equation}
  \label{eq:52}
  k^\mu = t^\mu -\hat \eta^\mu,
\end{equation}
is null future directed (since $t^\mu\eta_\mu=0$ ). Choosing $\xi^\mu=t^\mu$ and
$k^\mu$ given by (\ref{eq:52}), from (\ref{eq:32}) we obtain
\begin{equation}
  \label{eq:53}
  T_{\mu\nu} t^\mu t^\nu\geq T_{\mu\nu} t^\mu \hat \eta^\nu.
\end{equation}
Since $\eta$ is the square of the distance to the axis, we have that
\begin{equation}
  \label{eq:54}
  \eta \leq \R^2,
\end{equation}
where $\R$ is the radius a the ball that encloses $\dom$. Hence we deduce
\begin{align}
  \label{eq:55}
  J(\dom)=\frac{1}{c}\int_\dom T_{\mu\nu}t^\mu \eta^\nu 
  &=\frac{1}{c}\int_\dom T_{\mu\nu}t^\mu \eta^{1/2} \hat \eta^\nu ,  \\
&\leq \frac{\R}{c}\int_\dom T_{\mu\nu}t^\mu  \hat \eta^\nu,  \\
& \leq \frac{\R}{c}\int_\dom T_{\mu\nu}t^\mu  t^\nu,  \\
&= \frac{\R\E(\dom)}{c}.
\end{align}
Hence, we have proved  inequality (\ref{eq:27}) for a general energy
momentum tensor that satisfies the dominant energy condition (\ref{eq:32b}). 
Note, however, that we have not proved the rigidity statement as in theorem
\ref{t:ejr}. 

Finally, we prove inequality (\ref{eq:76}) for Electromagnetism in full generality. 
\begin{theorem}
\label{t:em}
Assume that $\rho(x,t_0)$, for some $t_0$, has compact support contained in
$\dom$.  Consider a solution of Maxwell's equations
(\ref{eq:32})--(\ref{eq:32v}) that decay at infinity. Then the following
inequality holds at $t_0$
\begin{equation}
    \label{eq:5c}
    \frac{c|J(\dom)|}{\R}+\frac{Q^2}{2\R} \leq  \E.
  \end{equation}
In particular, inequality  \eqref{eq:5c} implies
\begin{equation}
  \label{eq:69}
  \frac{Q^4}{4\R^2} + \frac{c^2 |J(\dom)|^2}{\R^2}\leq \E^2. 
\end{equation}
Moreover, if the equality in (\ref{eq:5c}) holds, then the electromagnetic field
is that  produced by a electrostatic spherical thin shell of radius $\R$ and
charge $Q$. For that case, the magnetic field vanished everywhere and hence
$J=0$.

\end{theorem}
\begin{proof}
  Consider the Coulomb gauge\footnote{I thank O. Reula for suggesting the idea
    of using the Coulomb gauge.}
  \begin{equation}
    \label{eq:35}
    B=\nabla \times \mathbf{A}, \quad \mathbf{E}=-\nabla \Phi -\frac{\partial
      \mathbf{A}}{\partial t}, 
  \end{equation}
where the potential $\mathbf{A}$ satisfies the Coulomb gauge condition
\begin{equation}
  \label{eq:45}
  \nabla \cdot \mathbf{A}=0. 
\end{equation}
In this gauge, the total energy can be written in the following form
\begin{align}
  \label{eq:49}
  \E &=\frac{1}{8\pi} \int_{\Rt} |\mathbf{E}|^2+ |\mathbf{B}|^2,\\
&= \frac{1}{8\pi} \int_{\Rt} |\nabla \Phi|^2 + 2 \nabla \Phi \cdot \frac{\partial
  \mathbf{A}}{\partial t}  + \left|\frac{\partial
  \mathbf{A}}{\partial t}\right|^2+ |\mathbf{B}|^2, \label{eq:49b}
\end{align}
where in line \eqref{eq:49b} we have used the expression \eqref{eq:35} for the
electric field in terms of the potential $\mathbf{A}$. 
For the second term in the integrand of (\ref{eq:49b}) we use the identity
\begin{align}
  \label{eq:51}
   \nabla \Phi \cdot \frac{\partial
  \mathbf{A}}{\partial t} & = \nabla \cdot \left(\Phi  \frac{\partial
  \mathbf{A}}{\partial t}  \right)- \Phi \frac{\partial \nabla \cdot \mathbf{A} }{\partial t},\\ 
&=\nabla \cdot \left(\Phi  \frac{\partial
  \mathbf{A}}{\partial t}  \right),
\label{eq:51b}
\end{align}
where in line \eqref{eq:51b} we have used the Coulomb gauge
condition (\ref{eq:45}). Using the  asymptotic falloff conditions for
$\Phi$ and $\mathbf{A}$ and Gauss theorem, from \eqref{eq:51b} we obtain
\begin{equation}
  \label{eq:56}
  \int_{\Rt} \nabla \Phi \cdot \frac{\partial
  \mathbf{A}}{\partial t} =0.
\end{equation}
Then, we have  the following expression for
the total energy
\begin{equation}
  \label{eq:21}
  \E = \frac{1}{8\pi} \int_{\Rt} |\nabla \Phi|^2 + \left|\frac{\partial
  \mathbf{A}}{\partial t}\right|^2+ |\mathbf{B}|^2.
\end{equation}

The potential $\Phi(x,t)$ satisfies the Poisson equation
\begin{equation}
  \label{eq:46}
  \Delta \Phi(x,t)=-4\pi \rho(x,t),
\end{equation}
for all $t$. At a fixed $t$, we can perform the same decomposition
(\ref{eq:58}) for the potential $\Phi(x,t)$ used in theorem \ref{t:1}. Then,
using equation (\ref{eq:18}), we
obtain
\begin{equation}
  \label{eq:22c}
  \E = \frac{Q^2}{2\R}+ \frac{1}{8\pi} \int_{\Rt} |\nabla \Phi_1|^2 + \left|\frac{\partial
  \mathbf{A}}{\partial t}\right|^2+ |\mathbf{B}|^2,  
\end{equation}
where $\Phi_1$ is defined by (\ref{eq:58}) and (\ref{eq:11}). 
By  the same integration by parts argument used to deduce (\ref{eq:56}) we obtain that
\begin{equation}
  \label{eq:56v}
  \int_{\Rt} \nabla \Phi_1 \cdot \frac{\partial
  \mathbf{A}}{\partial t} =0.
\end{equation}
Hence, we can write the energy (\ref{eq:22c}) in the following way
\begin{equation}
  \label{eq:65}
 \E = \frac{Q^2}{2\R}+ \frac{1}{8\pi} \int_{\Rt} \left |\nabla \Phi_1  + \frac{\partial
  \mathbf{A}}{\partial t}\right|^2 + |\mathbf{B}|^2. 
\end{equation}
We decompose the integral in (\ref{eq:65}) over the domains $\Rt\setminus \dom$ and $\dom$ and we use the following simple but important identity
\begin{align}
  \label{eq:50}
  \int_{\Rt} \left |\nabla \Phi_1  + \frac{\partial
  \mathbf{A}}{\partial t}\right|^2 & =  \int_{\Rt\setminus\dom} \left |\nabla \Phi_1  + \frac{\partial
  \mathbf{A}}{\partial t}\right|^2 + \int_{\dom} \left |\nabla \Phi_1  + \frac{\partial
  \mathbf{A}}{\partial t}\right|^2, \\
 & =  \int_{\Rt\setminus\dom} \left |\nabla \Phi_1  + \frac{\partial
  \mathbf{A}}{\partial t}\right|^2 + \int_{\dom} \left |\nabla \Phi  + \frac{\partial
  \mathbf{A}}{\partial t}\right|^2, \label{eq:50v} \\
& =  \int_{\Rt\setminus\dom} \left |\nabla \Phi_1  + \frac{\partial
  \mathbf{A}}{\partial t}\right|^2 + \int_{\dom} |\mathbf{E}|^2,  \label{eq:50vv}
\end{align}
where in line (\ref{eq:50v}) we have used that $\nabla \Phi_1=\nabla \Phi$ in
$\dom$ since $\Phi_0$ is constant in $\dom$. And in line \eqref{eq:50vv} we
have used the expression for the electric field in the Coulomb gauge
(\ref{eq:35}). Then we obtain the following expression for the energy $\E$
\begin{equation}
  \label{eq:78}
\E=\frac{Q^2}{2\R}+   \E(\dom) +  
\frac{1}{8\pi}  \int_{\Rt\setminus \dom} \left|\nabla \Phi_1+ 
\frac{\partial \mathbf{A}}{\partial t}\right|^2+|\mathbf{B}|^2,
\end{equation}
where $\E(\dom)$ is the electromagnetic energy density integrated over the
domain $\dom$, namely
\begin{equation}
  \label{eq:79}
\E(\dom)=    \frac{1}{8\pi} \int_{\dom}   |\mathbf{E} |^2 +|\mathbf{B}|^2.
\end{equation}
We use theorem \ref{t:ejr} to bound $\E(\dom)$ (i.e. the estimate
(\ref{eq:7})) and we finally have 
\begin{equation}
  \label{eq:20}
  \E-\frac{Q^2}{2\R}-\frac{c|J(\dom)|}{\R} \geq  \frac{1}{8\pi} \left(
  \int_{\Rt\setminus \dom} \left|\nabla \Phi_1  + \frac{\partial
  \mathbf{A}}{\partial t}\right|^2+ |\mathbf{B}|^2 +\int_\dom \left(1- \frac{|\mathbf{x} |}{\R}\right) \left(|\mathbf{E}|^2 +
 |\mathbf{B}|^2\right) \right).
\end{equation}
Since the left hand side of (\ref{eq:20}) is non-negative we have proved the
inequality \eqref{eq:5c}. Inequality (\ref{eq:20}) implies also the
rigidity statement. Assume the equality in \eqref{eq:5c}, then the integrand on
the right hand side of (\ref{eq:20}) should vanish.  This implies that
$\mathbf{B}=0$ everywhere, and hence the potential $\mathbf{A}$ is a
gradient. Using equation (\ref{eq:45}) and the falloff condition for
$\mathbf{A}$ we deduce that $\mathbf{A}=0$.  Then, using again (\ref{eq:20}) we
obtain that $\nabla \Phi_1=0$ and hence the statement is proved.

Taking the square of inequality \eqref{eq:5c},  we obtain 
\begin{equation}
  \label{eq:47}
 \frac{c|J| Q^2}{\R} + \frac{Q^4}{4\R^2}+ \frac{c^2 J^2}{\R^2}\leq \E^2,
\end{equation}
which, in particular,  implies inequality \eqref{eq:69}. 

\end{proof}

\section{General Relativity}
\label{sec:general-relativity}
In this section we study inequality (\ref{eq:76}) in General Relativity.
In section \ref{sec:scam} we discuss a remarkably relation between this
inequality and inequalities between size, charge and angular momentum. In
section \ref{sec:proofgr} we present a proof of inequality (\ref{eq:76}),
with $J=0$, for time-symmetric initial conditions.

\subsection{Inequalities between size, charge and angular
  momentum}
\label{sec:scam}

For a black hole the entropy is given by the horizon area $A$ 
\begin{equation}
  \label{eq:1v}
  S_{bh}  = \frac{k_Bc^3}{4G\hbar} A.
\end{equation}
Inequality \eqref{eq:6} is constructed in such a way that for a Kerr-Newman
black hole, using the formula \eqref{eq:1v}, we get an equality.  Moreover,
Szabados \cite{Szabados04} observed that for dynamical black holes this
inequality is also expected to hold. It is the generalization of the Penrose
inequality including charge and angular momentum (see the review article
\cite{Mars:2009cj} and \cite{Dain:2014xda} \cite{dain12} and the discussion
therein)

For ordinary bodies, inequality (\ref{eq:76}) is closely related to
inequalities between size, angular momentum and charge recently studied in
\cite{Dain:2013gma}. To show this relation we argue as follows.  The Hoop
conjecture essentially says that if matter is enclosed in a sufficiently small
region, then the system should collapse to a black hole \cite{thorne72}. Then
if the body is not a black hole we expect an inequality of the form
\begin{equation}
  \label{eq:37}
  \frac{G}{c^4} \E  \leq k\R,  
\end{equation}
where $k$ is an universal dimensionless constant of order one. The exact value
of $k$ will depend on the precise formulation of the Hoop conjecture and this
is not important in what follows. 

Using (\ref{eq:37}) to bound $\E$ in (\ref{eq:76})  we obtain
\begin{equation}
  \label{eq:2}
  \frac{Q^4}{4}+c^2 J^2 \leq  k^2 \frac{c^8}{G^2} \R^4.
\end{equation}
Note that the constant $G$ appears in (\ref{eq:2}). That is, inequality
(\ref{eq:2}) involves two fundamental constants ($c$ and $G$), in contrast to
(\ref{eq:76}) that involves only one ($c$).  On the other hand inequality
(\ref{eq:2}) involves fewer physical quantities (charge, angular momentum and
size) than inequality (\ref{eq:76}) (charge, angular momentum, size and
energy).

The bound \eqref{eq:2} implies 
\begin{equation}
  \label{eq:38}
\frac{G}{c^3} |J|  \leq k \R^2.
\end{equation}
Inequality \eqref{eq:38} was conjectured in \cite{Dain:2013gma} using
different kind of arguments as those leading to (\ref{eq:76}). With an
appropriate definition of size, a version of this inequality was proved for
constant density bodies in \cite{Dain:2013gma}. Recently Khuri
\cite{Khuri:2015zla} has proved it in the general case, using the same measure
of size as in \cite{Dain:2013gma}. However, these inequalities are not expected
to be sharp. We will came back to this point bellow.

Also from \eqref{eq:2} we get the inequality
\begin{equation}
  \label{eq:4}
  |Q| \leq (2k)^{1/2}\frac{c^2}{G^{1/2}} \R. 
\end{equation}
This inequality can be also deduced using similar arguments as in
\cite{Dain:2013gma} and it was studied for some particular examples in
\cite{Rubio14}. Recently Khuri \cite{Khuri:2015xpa} has proved a general
version of inequality \eqref{eq:4} using a similar (but not identical)
measure of size as the one used in inequality \eqref{eq:38}. As in the case
of angular momentum, this result is not expected to be sharp.

The relation between the Bekenstein bounds and inequalities \eqref{eq:38} and
\eqref{eq:4} provides two important new insights. The first one is the
following. We pointed out that in that inequality \eqref{eq:38} was conjectured
in \cite{Dain:2013gma} using heuristic physical arguments and also the
inequality \eqref{eq:4} can be deduced using similar kind of
arguments. However, with these arguments inequalities \eqref{eq:38} and
\eqref{eq:4} are deduced individually.  These kind of arguments do not seem to
provide a way of deducing the complete inequality \eqref{eq:2}, which is
obtained here for first time using the Bekenstein bounds.  Moreover, the
arguments presented above suggest that there is only one universal constant $k$
to be fixed. This constant can be fixed analyzing a simple limit case, for
example spherical symmetry, with $J=0$. We are currently working on this
problem \cite{anglada14}.

The second, and perhaps most important point concerns the rigidity of the
inequality \eqref{eq:2}. The arguments presented in \cite{Dain:2013gma} do not
give any insight about what happens when equality is reached in
\eqref{eq:2}. The Bekenstein bounds provide such statement. Let assume
that the equality is reached in \eqref{eq:2}. Since we have assumed that it is
not a black hole we can use the Hoop conjecture inequality \eqref{eq:37} to
obtain
\begin{equation}
  \label{eq:40}
   \frac{Q^4}{4}+c^2 J^2 =  k^2 \frac{c^8}{G^2} \R^4 \geq \E^2.
\end{equation}
But then we can use inequality \eqref{eq:76} to conclude that if the equality
is reached in \eqref{eq:2} then the equality should also hold in
\eqref{eq:76}. By the Bekenstein bound \eqref{eq:6}, this implies that the
entropy of the body is zero. Hence we have the following rigidity statement for
inequality \eqref{eq:2}: the equality is achieved if and only the entropy of
the body is zero. In General Relativity, this statement appears to imply that
in fact the equality is achieved if and only if the spacetime is flat. We will
further discuss this point in the next section.

\subsection{Proof of the inequality between charge, energy and size for time-symmetric initial data}
\label{sec:proofgr}
In General Relativity, the inequality (\ref{eq:2}) was proved in spherical
symmetry (in a different context) by Reiris \cite{Reiris:2014tva}.  In the
following we generalize this result to time-symmetric initial data.

The most important difficulty to study these kind of inequalities in curved
spaces is how to define the measure of size $\R$.  We propose a new measure of
size which is tailored to the proof of theorem \ref{t:GR}. This measure of size
represents a natural generalization to curved spaces of the definition
\ref{d:R} used in section \ref{sec:electromagnetism}. 

The definition of size and the proof of the theorem is based on the
\emph{inverse mean curvature flow} (IMCF).  A family of 2-surfaces on a
Riemannian manifold evolves under the IMCF if the outward normal speed at which
a point on the surface moves is given by the reciprocal of the mean curvature
of the surface. For the precise definition and properties of the IMCF we refer
to \cite{Huisken01}. The IMCF has played a key role in the  proof of the Riemannian Penrose inequality \cite{Huisken01}.

Using the IMCF we define the following radius $\R$ of a
region $\dom$ in a Riemannian manifold
\begin{definition}
\label{d:imcf}
  Consider a region $\dom$ on a complete, asymptotically flat, Riemannian
  manifold. Take a point $x_0$ on the manifold and consider the inverse mean
  curvature flow starting at this point. Consider the area of the first
  2-surface on the flow that encloses the region $\dom$, and define $\R_{x_0}$
  to be the area radius of this surface. The radius $\R$ of the region $\dom$
  is defined as the infimum of $\R_{x_0}$ over all points $x_0$ on the
  manifold.
\end{definition}
In figure \ref{fig:2} we draw an schematic picture of the flow starting at a
typical point $x_0$.  In flat space, the IMCF starting at a point develop
spheres, and hence the definition \ref{d:imcf} coincides with the definition
\ref{d:R} presented in the previous section. However, we emphasize that this
definition is very different as the one used in \cite{Dain:2013gma}
\cite{Khuri:2015zla} \cite{Khuri:2015xpa}. 

The radius $\R$ defined above certainly involves sophisticated mathematics,
however it is important to recall that it can be explicitly estimated
numerically for arbitrary curved backgrounds.

It is important to recall that, in general, the flow will develop
singularities. These singular behaviour can be treated using the weak
formulation discovered in \cite{Huisken01}. In what follows, for simplicity of
the presentation, we will assume that the flow is smooth, however all the
arguments are also valid in the weak formulation.

\begin{figure}
  \centering
  \includegraphics{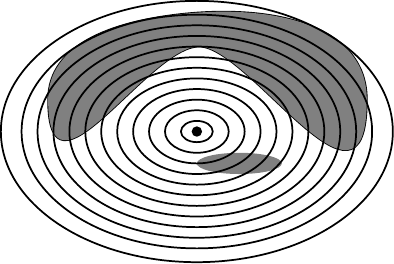}
  \caption{Schematic drawing of the inverse mean curvature flow from a typical
    point. The last surface is defined as the first one that enclosed the
    domain $\dom$.}
  \label{fig:2}
\end{figure}

We have the following result. 
\begin{theorem}
\label{t:GR}
Consider an asymptotically flat, complete, time-symmetric initial for
Einstein's equations that satisfy the dominant energy condition and with no
minimal surfaces.  Assume that there is a region $\dom$ outside of which the
data are electrovacuum.  Then we have
  \begin{equation}
    \label{eq:36}
 Q^2 \leq 2\E \R,
  \end{equation}
where $\E$ is the ADM mass, $Q$ is the charge contained in $\dom$ and $\R$ is
the radius of $\dom$ defined above. Moreover, if the equality in (\ref{eq:36})
holds, then the data is flat inside the region $\dom$.  
\end{theorem}
\begin{proof}
  The proof is inspired in Reiris's proof \cite{Reiris:2014tva} and it is a
  simple consequence of the results presented in \cite{Jang79} and
  \cite{Huisken01}.

The crucial property of the IMCF is the Geroch monotonicity of the Hawking
energy.  The Hawking energy of a closed 2-surface $\Su$ is given by
\begin{equation}
  \label{eq:34}
  \EH(\Su)= \sqrt{\frac{A}{16\pi}}\left( 1- \frac{1}{16\pi}\int_{\Su} H^2 
  \right),
\end{equation}
where $H$ is the mean curvature of the surface and $A$ its area.  The Geroch
monotonicity can be written in the following form. Assume that the flow runs
between a surface $\Su_r$ and a surface $\Su_s$, with $r<s$, then we have
  \begin{equation}
    \label{eq:41}
    \EH(\Su_s) \geq \EH(\Su_r) + \frac{1}{(16\pi)^{3/2}} \int_r^s (A_t)^{1/2}
    \int_{\Su_t}   R\, dt,
  \end{equation}
where $R$ is the scalar curvature. Note that  the dominant energy condition for
time-symmetric data implies that
$R\geq 0$.  We will use inequality (\ref{eq:41}) in
two steps. 

First, consider an arbitrary point $x_0$ and run the IMCF from $x_0$. Since the
data satisfy the dominant energy condition, an small sphere around $x_0$ has
non-negative Hawking mass. Moreover, the assumption that there are no minimal
surfaces on the data guarantees that the flow runs up to infinity (even in the presence of
singularities, see \cite{Huisken01}). Then, using (\ref{eq:41}) we conclude
that any level set of the flow has non-negative Hawking energy. In particular,
the surface $\Su_{x_0}$ that encloses the region $\dom$ used in the definition
\ref{d:imcf}, that is
\begin{equation}
  \label{eq:43}
  \EH(\Su_{x_0}) \geq 0.
\end{equation}
Denote by $A_{x_0}$ the area of $\Su_{x_0}$ and  the area radius is given by
$\R_{x_0}=\sqrt{A_{x_0}/4\pi}$.

In the second step, we continue the flow from the surface $\Su_{x_0}$ to
infinity. Following  \cite{Jang79}, we bound the integral of the scalar
curvature in terms of the charge
\begin{equation}
  \label{eq:77}
  \frac{1}{(16\pi)^{3/2}} \int_{x_0}^\infty (A_t)^{1/2}
    \int_{\Su_t}   R  \, dt \geq \frac{Q^2}{2\R_{x_0}},
\end{equation}
where we have used that the charge is conserved outside $\Su_{x_0}$, since by
construction $\Su_{x_0}$ encloses the region $\dom$ and by assumption the
support of the charge density is contained in $\dom$. Using (\ref{eq:77}) and
(\ref{eq:41}) we obtain 
\begin{equation}
  \label{eq:42}
  \E-\frac{Q^2}{2\R_{x_0}} \geq  \EH(\Su_{x_0}).
\end{equation}
Using (\ref{eq:43}) we finally get 
\begin{equation}
  \label{eq:48}
  \E-\frac{Q^2}{2\R_{x_0}}\geq 0. 
\end{equation}
In particular this inequality applies to the radius $\R$ and hence inequality
(\ref{eq:36}) follows. 

If the equality holds in (\ref{eq:42}), then we have $\EH(\Su_{x_0})=0$ and
hence we can use the same rigidity argument as in \cite{Huisken01} to conclude
that inside $\Su_{x_0}$ the data are flat. 
\end{proof}

We have obtained a similar kind of estimate as in the electromagnetic
case \eqref{eq:12} in which $\EH(\Su_0)$ is interpreted as the quasilocal
energy inside $\Su_0$.

Comparing theorem \ref{t:GR} with theorem \ref{t:em} in Electromagnetism, we
see that there is no rigidity statement outside the region $\dom$ in theorem
\ref{t:GR}. The natural question is whether a similar statement as in theorem
\ref{t:em} holds, namely, the equality implies that the field is produced by a
charged thin shell. However, it is likely that the charged thin shell in
General Relativity never saturate the inequality (in contrast with
Electromagnetism). The reason is that the rest energy of the shell is now taken
into account. Hence, a stronger rigidity statement is expected for theorem
\ref{t:GR}: the equality holds if and only if the complete data are flat. We
are currently working on this problem \cite{anglada14}.

It will interesting to include angular momentum in theorem \ref{t:GR}. However,
this appears to be a difficult problem. In particular it is not clear how to
include angular momentum in the inequality using the IMCF.

\section{Acknowledgements}
This work was inspired by a talk of Horacio Casini, in the conference
Strings@ar, La Plata, Argentina, November 2014.  This work was supported by
grant PIP of CONICET (Argentina) and grant Secyt-UNC (Argentina).


\end{document}